\newif\iflncs
\newif\iflipics

\lncsfalse
\lipicsfalse

\iflncs
	\documentclass[envcountsame]{llncs}
\else
 \iflipics
	\documentclass[a4paper]{lipics-authors/lipics}
	\graphicspath{{./lipics-authors/}}
 \else
	\documentclass[11pt]{article}
 \fi
\fi

\usepackage{amsmath}
\usepackage{amssymb}
\usepackage{verbatim}
\usepackage{graphicx}
\usepackage{xspace}
\usepackage{lmodern}
\usepackage{tikz}
\usepackage{todonotes}
\usetikzlibrary{arrows}
\usepackage{cite}

\usepackage{authblk}

\usepackage{algorithm}

\iflipics 

\else
	\iflncs
	\else
		\usepackage[paper=letterpaper, margin=1in]{geometry}
	\fi
\fi

\iflncs
	\spnewtheorem{observation}[lemma]{Observation}{\bfseries}{\itshape}
	\newtheorem{definition}[theorem]{Definition}
	
\else
	\usepackage{amsthm}

	\iflipics
		\theoremstyle{plain}
		\newtheorem{observation}[theorem]{Observation}
	\else
		\newtheorem{observation}{Observation}
		\newtheorem{theorem}{Theorem}
		\newtheorem{lemma}{Lemma}
		\newtheorem{corollary}[lemma]{Corollary}

		\newtheorem{definition}[theorem]{Definition}
		\usepackage[labelfont=bf,small,margin=0cm]{caption}
		
	\fi
\fi

\usepackage{ifthen}

\ifthenelse{\isundefined{\GenerateShortVersion}}
{
\def\LongVersion{}
\def\LongVersionEnd{}
\long\def\ShortVersion#1\ShortVersionEnd{}
}
{
\def\ShortVersion{}
\def\ShortVersionEnd{}
\long\def\LongVersion#1\LongVersionEnd{}
}

\LongVersion 

\LongVersionEnd 

\ShortVersion 
\ShortVersionEnd 

\ShortVersion 
\usepackage{times}
\ShortVersionEnd 

\ShortVersion 
\renewcommand{\paragraph}[1]{\par\noindent\textbf{#1}}
\ShortVersionEnd 

\newcommand{\Sect}{Sec.}

\makeatletter
\newcommand\ftime{\ensuremath{t^{f}\@ifnextchar({{}^{\!}}{}}\xspace}
\makeatother

\newenvironment{IntuitionSpotlight}[0]
{\par%
\setlength{\leftskip}{0.5\parindent}%
\setlength{\rightskip}{0.5\parindent}%
\noindent\itshape\textbf{Intuition spotlight:}}
{\par\ignorespacesafterend} 

\newif\ifShortVersion
\ShortVersionfalse

\usepackage{collect}
\usepackage{xparse}

\definecollection{Proofs}
\ifShortVersion
	\makeatletter
	\NewDocumentEnvironment{proof*}{mo}{%
		\@nameuse{collect}{Proofs}{%
			\begin{proof}[#1]
		}{
			\end{proof}
		}
	}{%
		\@nameuse{endcollect}%
	}
	\makeatother
\else
	\NewDocumentEnvironment{proof*}{mo}{%
		\IfValueTF{#2}{\begin{proof}[#2]}{\begin{proof}}
	}{%
		\end{proof}
	}
\fi

\usepackage{collector}

\ifShortVersion
	\renewcommand\paragraph[1]{\vspace{2pt}\noindent \textbf{#1}}
\else
\fi

\makeatletter
\def\lastabbrevdot{\@ifnextchar.{}{\@ifnextchar,{.}{.\ }}}
\makeatother
\def\whp{w.h.p\lastabbrevdot}

\title{Dynamic Networks of Finite State Machines\footnote{This work was partially supported by ERC Grant No. 336495 (ACDC). Licensed under CC-BY-NC-ND 4.0}}
\iflncs
	\author{Yuval Emek \and Jara Uitto}
	\institute{ETH Z\"urich, Switzerland}
\else
	\date{}
	\author[1]{Yuval Emek}
	\author[2,3]{Jara Uitto}
	\affil[1]{Technion, Israel}
	\affil[2]{ETH Z\"urich, Switzerland}
	\affil[3]{University of Freiburg, Germany}
	\iflipics
		\authorrunning{Y. Emek and J. Uitto}
		\Copyright{Jara Uitto, Yuval Emek}

		\serieslogo{}
		\volumeinfo
		  {}
		  {0}
		  {The 23rd International Colloquium on Structural Information and Communication Complexity (SIROCCO)}
		  {1}
		  {1}
		  {1}
		\EventShortName{}
		\DOI{10.4230/LIPIcs.xxx.yyy.p}
	\fi
\fi

\newcommand{\bigO}{\mathcal{O}}

\newcommand{\Expectation}[0]{\mathbb{E}}
\newcommand{\Probability}[0]{\mathbb{P}}

\newcommand{\Yuvals}{\texttt{Proportional}}

\newcommand{\greedy}{\texttt{Greedy}}
\newcommand{\nbh}{\Gamma}
\newcommand{\hnbh}{\hat \Gamma}
\newcommand{\pkaks}{\mathcal{E}}

\begin{document}

\iflipics
\else
 \begin{titlepage}
 
\fi

\maketitle
\iflipics
\else
 \thispagestyle{empty}
\fi
\begin{abstract}
Like distributed systems, biological multicellular processes are subject to
dynamic changes and a biological system will not pass the
survival-of-the-fittest test unless it exhibits certain features that enable
fast recovery from these changes.
In most cases, the types of dynamic changes a biological process may
experience and its desired recovery features differ from those traditionally
studied in the distributed computing literature.
In particular, a question seldomly asked in the context of distributed digital
systems and that is crucial in the context of biological cellular networks, is
whether the system can keep the changing components \emph{confined} so that
only nodes in their vicinity may be affected by the changes, but nodes
sufficiently far away from any changing component remain unaffected.

Based on this notion of confinement, we propose a new metric for measuring the
dynamic changes recovery performance in distributed network algorithms
operating under the \emph{Stone Age} model (Emek \& Wattenhofer, PODC 2013),
where the class of dynamic topology changes we consider includes
inserting/deleting an edge, deleting a node together with its incident edges,
and inserting a new isolated node.
Our main technical contribution is a distributed algorithm for maximal
independent set (MIS) in synchronous networks subject to these
topology changes that performs well in terms of the aforementioned new metric.
Specifically, our algorithm guarantees that nodes which do not experience a
topology change in their immediate vicinity are not affected and that
all surviving nodes (including the affected ones) perform
$\bigO((C + 1) \log^{2} n)$
computationally-meaningful steps, where $C$ is the number of topology
changes;
in other words, each surviving node performs
$\bigO(\log^{2} n)$
steps when amortized over the number of topology changes.
This is accompanied by a simple example demonstrating that the linear
dependency on $C$ cannot be avoided.
\end{abstract}
 \iflipics
 \else


	
	\end{titlepage}
 \fi

\section{Introduction}
\label{section:introduction}

The biological form of close-range (juxtacrine) message passing relies on
designated messenger molecules that bind to crossmembrane receptors in
neighboring cells;
this binding action triggers a signaling cascade that eventually affects gene
expression, thus modifying the neighboring cells' states.
This mechanism should feel familiar to members of the distributed computing
community as it resembles the message passing schemes of distributed digital
systems.
In contrast to nodes in distributed digital systems, however, biological cells
are not believed to be Turing complete, rather each biological cell is pretty
limited in computation as well as communication.
In attempt to cope with these differences, Emek and
Wattenhofer~\cite{Emek2013} introduced the \emph{Stone Age} model of
distributed computing (a.k.a.\ networked finite state machines), where each
node in the network is a very weak computational unit with limited
communication capabilities and showed that several fundamental distributed
computing problems can be solved efficiently under this model.

An important topic left outside the scope of \cite{Emek2013} is that of
\emph{dynamic topology changes}.
Just like distributed digital systems, biological systems may experience
local changes and the ability of the system to recover from these
changes is crucial to its survival.
However, the desired recovery features in biological cellular networks
typically differ from those traditionally studied in the distributed computing
literature.
In particular, a major issue in the context of biological cellular networks,
that is rarely addressed in the study of distributed digital systems, is that
of \emph{confining} the topology changes:
while nodes in the immediate vicinity of a topology change are doomed to be
affected by it (hopefully, to a bounded extent), isolating the nodes
sufficiently far away from any topology change so that their operation remains
unaffected, is often critical.
Indeed, a biological multicellular system with limited energy resources cannot
afford every cell division (or death) to have far reaching effects on the
cellular network.


In this paper, we make a step towards bringing the models from computer
science closer to biology by extending the Stone Age model to accommodate four
types of dynamic topology changes:
(1) deleting an existing edge;
(2) inserting a new edge;
(3) deleting an existing node together with its incident edges; and
(4) inserting a new isolated node.
We also introduce a new method for measuring the performance of a network in
recovering from these types of topology changes that takes into account the
aforementioned confinement property.
This new method measures the number of ``computationally-meaningful''
steps made by the individual nodes, which are essentially all steps in which
the node participates (in the weakest possible sense) in the global
computational process.
An algorithm is said to be \emph{effectively confining} if
(i) the runtime of the nodes that are not adjacent to any topology change is
$\log^{\bigO (1)} n$; and
(ii) the global runtime (including all surviving nodes) is
$(C + 1) \log^{\bigO (1)} n$,
where $C$ is the number of topology changes throughout the execution.
In other words, the global runtime is
$\log^{\bigO (1)} n$
when amortized over the number of changes. 

Following that, we turn our attention to the extensively studied \emph{maximal
independent set (MIS)} problem and design a randomized effectively confining
algorithm for it under the Stone Age model extended to dynamic topology
changes.
This is achieved by carefully augmenting the MIS algorithm introduced in
\cite{Emek2013} with new components, tailored to ensure fast recovery from
topology changes.
Being a first step in the study of recovery from dynamic changes under the
Stone Age model, our algorithm assumes a \emph{synchronous} environment and it
remains an open question whether this assumption can be lifted.
Nevertheless, this assumption is justified by the findings of Fisher et
al.~\cite{FisherHMP2008} that model cellular networks as being subject to a
\emph{bounded asynchrony} scheduler, which is equivalent to a synchronous
environment from an algorithmic perspective.

\paragraph{Paper's Organization}
An extension of the Stone Age model of~\cite{Emek2013} to dynamic topology
changes is presented in \Sect{}~\ref{section:model} together with our new
method for evaluating the recovery performance of distributed algorithms.
In \Sect{}~\ref{sec: MISalgo}, we describe the details of our MIS algorithm.
Then, in \Sect{}~\ref{sec: analysisProportional}, we show that each node not
affected by a topology change will reach an output state in $\bigO(\log^2 n)$
rounds.
In \Sect{} \ref{sec: quality} and in \Sect{} \ref{sec: analysisGreedy}, we finish the
analysis of our MIS algorithm by establishing an
$\bigO((C + 1) \log^2 n)$
upper bound on the global runtime.
The runtime of the new MIS algorithm is shown to be near-optimal in
\Sect{}~\ref{sec: lower} by proving that the global runtime of any algorithm
is $\Omega(C)$.
In \Sect{}~\ref{sec: logDist}, we show that the runtime of any node
$u$ can be further bounded by
$\bigO((C_u + 1) \log^2 n)$,
where $C_u$ is the number of topology changes that occurred within
$\bigO(\log n)$
hops from $u$.
Some interesting open questions are listed in
\Sect{}~\ref{section:conclusions}.

\paragraph{Related Work}
The standard model for a network of communicating devices is the message
passing model \cite{PelegBook,Linial1992}.
There are several variants of this model, where the power of the network has
been weakened.
Perhaps the best-known variant of the message passing model is the congest
model~\cite{PelegBook}, where the message size is limited to size logarithmic in
the size of the input graph.
A step to weaken the model further is to consider interference of messages,
i.e., a node only hears a message if it receives a single message per round
--- cf.\ the \emph{radio network model}~\cite{ChlamtacKuten85}.
In the \emph{beeping model}~\cite{Flury2010Slotted, CornejoKuhn10}, the
communication capabilities are reduced further by only allowing to send
beeps that do not carry information, where a listening node cannot
distinguish between a single beep and multiple beeps transmitted by its
neighbors.

The models mentioned above focus on limiting the communication but not the
computation, i.e., the nodes are assumed to be strong enough to perform
unlimited (local) Turing computations in each round.
Networks of nodes weaker than Turing machines have been extensively studied in
the context of \emph{cellular automata}~\cite{vonNeumann66, Gardner70,
Wolfram02}.
While the cellular automata model typically considers the deployment of finite
state machines in highly regular graph structures such as the grid, the
question of modeling cellular automata in arbitrary graphs was tackled by Marr
and H\"{u}tt in~\cite{Marr2009}, where a node changes its binary state
according to the densities of its neighboring states.
Another extensively studied model for distributed computing in a network of
finite state machines is \emph{population protocols} \cite{AngluinADFP2006}
(see also \cite{AspnesR2009, MichailCS11}), where the nodes communicate through
a series of pairwise rendezvous.
Refer to \cite{Emek2013} for a comprehensive account of the similarities and
differences between the Stone Age model and the models of cellular automata
and population protocols.

Distributed computing in dynamic networks has been extensively
studied~\cite{APPS92, Walter1998, Li2004, Kuhn2005, Hayes2009}.
A classic result by Awerbuch and Sipser states that under the message
passing model, any algorithm designed to run in a static network can be
transformed into an algorithm that runs in a dynamic network with only a
constant multiplicative runtime overhead~\cite{Awerbuch1988}.
However, the transformation of Awerbuch and Sipser requires storing the whole
execution history and sending it around the network, which is not possible
under the Stone Age model.
Some dynamic network papers rely on the assumption that the topology changes
are spaced in time so that the system has an opportunity to recover before
another change occurs~\cite{Korman2008, Malpani2003, Konig2013}.
The current paper does not make this assumption.

The \emph{maximal independent set (MIS)} problem has a long history in the
context of distributed algorithms~\cite{Valiant82, Luby86, ABI86, Cole1986,
Linial1992, Panconesi1996}.
Arguably the most significant breakthrough in the study of message passing MIS
algorithms was the $\bigO(\log n)$ algorithm of Luby~\cite{Luby86} (developed
independently also by Alon et al.~\cite{ABI86}).
Later, Barenboim et al.~\cite{Barenboim2012} showed an upper bound of $2^{\bigO\left(\sqrt{\log \log n}\right)}$ in the case of polylogarithmic maximum degree and an $\bigO(\log \Delta + \sqrt{\log n})$ bound for general graphs.
For growth bounded graphs, it was shown by Schneider et
al.~\cite{sw10} that MIS can be computed in
$\bigO(\log^* n)$
time.
In a recent work, Ghaffari studied the \emph{local} complexity of computing an MIS, where the time complexity is measured from the perspective of a single node, instead of the whole network~\cite{Ghaffari2016}. 
In a similar spirit, we provide, in addition to a global runtime bound, a runtime analysis from the perspective of a single node.

In the radio networks realm, with $f$ channels,
an MIS can be computed in
$\Theta(\log^2 n / f) + \tilde{\bigO}(\log n)$
time~\cite{Daum2013}, where $\tilde{\bigO}$ hides factors polynomial in $\log \log n$.
The MIS problem was extensively studied also under the beeping
model~\cite{AfekABHBB2011, AfekABCHK2011, ScottJX2013}.
Afek et al.~\cite{AfekABCHK2011} proved that if the nodes are provided with an
upper bound on $n$ or if they are given a common sense of time, then the MIS
problem can be solved in
$\bigO(\log^{\bigO(1)} n)$ time.
This was improved to
$\bigO (\log n)$
by Scott et al.~\cite{ScottJX2013} assuming that the nodes can detect sender collision

On the negative side, the seminal work of Linial~\cite{Linial1992} provides
a runtime lower bound of
$\Omega(\log^* n)$
for computing an MIS in an $n$-node ring under the message passing model~\cite{Linial1992}.
Kuhn et al.~\cite{Kuhn2016} established a stronger lower bound for general graphs stating that it
takes
\[
	\min { \Omega\left(\sqrt{\frac{\log n}{\log \log n}}, \frac{\log \Delta}{\log \log \Delta}\right) }
\]
rounds to compute an MIS, where $\Delta$ is the maximum degree of the input
graph.
For uniform algorithms in radio networks (and therefore, also for the beeping
model) with asynchronous wake up schedule, there exists a lower bound of
$\Omega(\sqrt{n / \log n})$
communication rounds~\cite{AfekABCHK2011}.

Containing faults within a small radius of the faulty node has been studied in
the context of \emph{self-stabilization}~\cite{Ghosh1996}.
An elegant MIS algorithm was developed under the assumption that the
activation times of the nodes are controlled by a \emph{central daemon}
who activates the nodes one at a time~\cite{Shukla1995, Lin2003}.
In contrast, we follow the common assumption that all nodes are activated
simultaneously in each round.
In the self-stabilization realm, the performance of an algorithm is typically
measured as a function of some network parameter, such as the size of
the network or the maximum degree, whereas in the current paper, the
performance depends also on the number of failures.

With respect to the performance evaluation, perhaps the works closest to ours
are by Kutten and Peleg~\cite{Kutten1999, Kutten2000}, where the concepts of
\emph{mending algorithms} and \emph{tight fault locality} are introduced.
The idea behind a fault local mending algorithm is to be able to recover
a legal state of the network after a fault occurs, measuring the performance
in terms of the number of faults.
The term tight fault locality reflects the property that an algorithm running
in time $\bigO(T(n))$ without faults is able to mend the network in time
$\bigO(T(F))$,
where $F$ denotes the number of faults.
The algorithm of Kutten and Peleg recovers an MIS in time $\bigO(\log F)$, but
they use techniques that require nodes to count beyond constant numbers, which
is not possible in the Stone Age model.
Furthermore, they consider transient faults, whereas we consider permanent
changes in the network topology.


\section{Model}
\label{section:model}
Consider some network represented by an undirected graph $G = (V, E)$, where
the nodes in $V$ correspond to the network's computational units and the edges
represent bidirectional communication channels.
Adopting the (static) Stone Age model of~\cite{Emek2013}, each node $v \in V$
runs an algorithm captured by the $8$-tuple
$\Pi = \left\langle Q, q_{0}, q_{yes}, q_{no}, \Sigma, \sigma_0, b, \delta
\right\rangle$,
where $Q$ is a fixed set of \emph{states};
$q_{0} \in Q$ is the \emph{initial state} in which all nodes reside when they
wake up for the first time;
$q_{yes}$ and $q_{no}$ are the \emph{output states}, where the former (resp.,
latter) represents membership (resp., non-membership) in the output MIS;
$\Sigma$ is a fixed \emph{communication alphabet};
$\sigma_0 \in \Sigma$ is the \emph{initial letter};
$b \in \mathbb{Z}_{>0}$ is a \emph{bounding parameter}; and
$\delta : Q \times \{ 0, 1, \dots, b \}^{|\Sigma|} \rightarrow 2^{Q \times
(\Sigma \cup \{\varepsilon\})}$
is the \emph{transition function}.
For convenience, we sometimes denote a state transition from $q$ to $q'$ by $(q \rightarrow q')$ and omit the rules associated with this transition from the notation.
\footnote{%
For simplicity, the definition presented here is already tailored to the MIS problem. For a general problem $P$, the initial state $q_{0}$ can be replaced with a set of initial states, representing the possible inputs of $P$, and we can have as many output states as needed to represent the possible outputs of $P$ (captured in the MIS case by the two states $q_{yes}$ and $q_{no}$).}

Node $v$ communicates with its neighbors by transmitting messages that consist
of a single letter $\sigma \in \Sigma$ such that the same letter $\sigma$ is
sent to all neighbors.
It is assumed that $v$ holds a port $\phi_u(v)$ for each neighbor
$u$ of $v$ in which the last message (a letter in $\Sigma$) received from $u$
is stored.
Transmitting the designated empty symbol $\varepsilon$ corresponds to the case
where $u$ does not transmit any message.
In other words, when node $u$ transmits the $\varepsilon$ letter, the letters in ports $\phi_u(v)$, for all neighbors $v$ of $u$, remain unchanged.
In the beginning of the execution, all ports contain the initial letter $\sigma_0$.

The execution of the algorithm proceeds in discrete synchronous rounds indexed
by the positive integers.
In each round $r \in \mathbb{Z}_{> 0}$, node $v$ is in some state $q \in Q$.
Let
$\sharp(\sigma)$
be the number of appearances of the letter
$\sigma \in \Sigma$
in $v$'s ports in round $r$ and let
$\langle \min \{ \sharp(\sigma), b \} \rangle_{\sigma \in \Sigma}$
be a $\Sigma$-indexed vector whose $\sigma$-entry is set to the minimum
between $\sharp(\sigma)$ and the bounding parameter $b$.
Then the state $q'$ in which $v$ resides in round $r + 1$ and the message
$\sigma'$ that $v$ sends in round $r$ (appears
in the corresponding ports of $v$'s neighbors in round $r + 1$ unless
$\sigma' = \varepsilon$) are chosen uniformly at random among the pairs in 
\[
\delta \left( q, \langle \min \{ \sharp(\sigma), b \} \rangle_{\sigma
\in \Sigma} \right)
\subseteq
Q \times (\Sigma \cup \{ \varepsilon \}) \, .
\]
This means that $v$ tosses an unbiased die with
\[
	|\delta(q, \langle \min \{ \sharp(\sigma), b \} \rangle_{\sigma \in \Sigma})|
\]
faces when deciding on $q'$ and $\sigma'$.

To ensure that our state transitions are well defined, we require that for all $q \in Q$ and $x \in \{ 0, 1, \dots, b \}^{|\Sigma|}$ it holds that
$|\delta(q, x)| \geq 1$
and say that a state transition is \textit{deterministic} if
$|\delta(q, x)| = 1$.

\paragraph{Topology changes}
In contrast to the model presented in \cite{Emek2013}, our network model
supports dynamic \emph{topology changes} that belong to the following four
classes:
\begin{enumerate}

\item
\emph{Edge deletion}:
Remove a selected edge $e = \{u, v\}$ from the current graph.
The corresponding ports $\phi_u(v)$ and $\phi_v(u)$ are removed and the
messages stored in them are erased.

\item
\emph{Edge insertion}:
Add an edge connecting nodes $u, v \in V$ to the current graph.
New ports $\phi_u(v)$ and $\phi_v(u)$ are introduced storing the initial
letter $\sigma_{0} \in \Sigma$.

\item
\emph{Node deletion}:
Remove a selected node $v \in V$ from the current graph with all its incident
edges.
The corresponding ports $\phi_v(u)$ of all $v$'s neighbors $u$ are eliminated
and the messages stored in them are erased.

\item
\emph{Node insertion}:
Add a new isolated (i.e., with no neighbors) node to the current graph.
Initially, the node resides in the initial state $q_{0}$.

\end{enumerate}
We assume that the schedule of these topology changes is controlled by an
\textit{oblivious adversary}.
Formally, the strategy of the adversary associates a (possibly empty) set of
topology changes with each round
$r \in \mathbb{Z}_{> 0}$
of the execution so that the total number of changes throughout the execution,
denoted by $C$, is finite.
This strategy may depend on the algorithm $\Pi$, but not on the random choices
made during the execution.

To be precise, each round is divided into four successive steps as follows:
$(i)$
messages arrive at their destination ports;
$(ii)$
topology changes occur;
$(iii)$
the transition function is applied; and
$(iv)$
messages are transmitted.
In particular, a message transmitted by node $v$ in round $r$ will not be read
by node $u$ in round $r + 1$ if edge $\{ u, v \}$ is inserted or deleted in
round $r + 1$.
%
%
%
It is convenient to define the \emph{adversarial graph sequence}
$\mathcal{G} = G_{1}, G_{2}, \dots$
so that
$G_{1}$ is the initial graph and
$G_{r + 1}$ is the graph obtained from $G_{r}$ by applying to it the topology
changes scheduled for round $r$.
By definition, $\mathcal{G}$ is fully determined by the initial graph and the
adversarial policy (and vice versa).
The requirement that $C$ is finite implies, in particular, that $\mathcal{G}$
admits an infinite suffix of identical graphs.
Let $n$ be the largest number of nodes that co-existed in the same round, i.e.,
$n = \max_{r} |V(G_{r})|$,
where $V(G_{r})$ is the node set of $G_{r}$.

\paragraph{Correctness}
We say that node $u$ resides in state $q$ in round $r$ if the state of $u$ is
$q$ at the beginning of round $r$.
The algorithm is said to be in an \textit{output configuration} in round $r$
if every node $u \in G_r$ resides in an output state ($q_{yes}$ or $q_{no}$).
The output configuration is said to be \emph{correct} if the states of the
nodes (treated as their output) correspond to a valid MIS of the current
graph $G_{r}$.
An algorithm $\Pi$ is said to be \emph{correct} if the following conditions
are satisfied for every adversarial graph sequence:
(C1)
If $\Pi$ is in a non-output configuration in round $r$, then it will move
to an output configuration in some (finite) round $r' > r$ w.p.\
$1$.\footnote{%
Throughout, we use w.p.\ to abbreviate ``with probability'' and w.h.p.\ to
abbreviate ``with high probability'', i.e., with probability $n^{-c}$ for any
constant $c$.}
(C2)
If $\Pi$ is in an output configuration in round $r$, then this output
configuration is correct (with respect to $G_{r}$).
(C3)
If $\Pi$ is in an output configuration in round $r$ and
$G_{r + 1} = G_{r}$,
then $\Pi$ remains in the same output configuration in round $r + 1$.

\paragraph{Restrictions on the Output States}
Under the model of \cite{Emek2013}, the nodes are not allowed to change their
output once they have entered an output state.
On the other hand, this model allows for multiple output states that
correspond to the same problem output (``yes'' and ``no'' in the MIS case).
In other words, it is required that the output states that
correspond to each problem output form a sink of the transition function.
Since our model accommodates dynamic topology changes which might turn a
correct output configuration into an incorrect one, we lift this restriction:
our model allows transitions from an output state to a non-output state, thus
providing the algorithm designer with the possibility to escape output
configurations that become incorrect.
Nevertheless, to prevent nodes in an output state from taking any
\emph{meaningful} part in the computation process, we introduce the following
new (with respect to \cite{Emek2013}) restrictions:
(1) each possible output is represented by a unique output state;
(2) a transition from an output state to itself is never accompanied by a
letter transmission (i.e., it transmits $\varepsilon$);
(3) all transitions originating from an output state must be deterministic; and
(4) a transition from an output state must lead either to itself or to a
non-output state.

\paragraph{Runtime}
Fix some adversarial graph sequence
$\mathcal{G} = G_{1}, G_{2}, \dots$
and let $\eta$ be an execution of a correct algorithm $\Pi$ on $\mathcal{G}$
(determined by the random choices of $\Pi$), where the execution corresponds to a sequence of global configurations of the system.
Round $r$ in $\eta$ is said to be \textit{silent} if $\Pi$ is in a
correct output configuration and no topology change occurs in round $r$.
The \textit{global runtime} of $\eta$ is the number of non-silent rounds.
Node $v$ is said to be \emph{active} in round $r$ of $\eta$ if it resides in a
non-output state, i.e., some state in $Q - \{ q_{yes}, q_{no} \}$.
The \emph{(local) runtime} of $v$ in $\eta$ is the number of rounds in
which $v$ is active.

Let $N_{r}(v)$ be the (inclusive) neighborhood of $v$ and let $E_r(v)$ be the edges incident to $v$ in $G_{r}$.
Node $v$ is said to be \emph{affected} under $\mathcal{G}$ if either $v$ or
one of its neighbors experienced an edge insertion/deletion, that is, there
exists some round $r$ and some node $u \in N_{r}(v)$ such that $E_r(u)$ is not identical to $E_{r + 1}(u)$.
Algorithm $\Pi$ is said to be \emph{effectively confining} if the following
conditions are satisfied for every adversarial graph sequence $\mathcal{G}$:
(1)
the expected maximum local runtime of the non-affected nodes is
$\log^{\bigO(1)} n$; and
(2)
the expected global runtime is $(C + 1) \cdot \log^{\bigO(1)} n$,
namely, $\log^{\bigO(1)} n$ when amortized over the number of topology
changes.
Notice that the bound on the global runtime directly implies the same bound on
the local runtime of any affected node.


\newcommand{\WIN}{\textrm{WIN}}
\section{An MIS Algorithm}
\label{sec: MISalgo}
Our main goal is to design an algorithm under the Stone Age model for the maximal independent set (MIS) problem that is able to tolerate topology changes. 
For a graph $G = (V, E)$, a set of nodes $I \subseteq V$ is independent if for all $u, v \in I$, $\{u, v\} \notin E$. 
An independent set $I$ is maximal if there is no other set $I' \subseteq V$ such that $I \subset I'$ and $I'$ is independent. 

Following the terminology introduced in the model section, we show that our algorithm is effectively confining. In \Sect{}~\ref{sec: lower}, we provide a straightforward lower bound example that shows that under our model, our solution is within a polylogarithmic factor from optimal. In other words, the linear dependency on the number of changes is inevitable.
Throughout, the bounding parameter of our algorithm is one.
This indicates that a node is able to distinguish between the cases of having a single appearance of a letter $\sigma$ in some of its ports and not having any appearances of $\sigma$ in any of its ports.

The basic idea behind our algorithm is that we first use techniques from~\cite{Emek2013} to come up with an MIS quickly and then we fix any errors that the topology changes might induce. 
In other words, our goal is to first come up with a \textit{proportional} MIS, where the likelihood of a node to join the MIS is inversely proportional to the number of neighbors the node has that have not yet decided their output. 
We partition the state set of our algorithm into two components. One of the components contains the input state and is responsible for computing the proportional MIS. 
Once a node has reached an output state or detects that it has been affected by a topology change, it transitions to the second component, which is responsible for fixing the errors, and never enters an active state of the first component. 
One of the reasons behind dividing our algorithm into two seemingly similar components is that the first component has stronger runtime guarantees, but requires that the nodes start in a ``nice'' configuration, whereas the second component does not have this requirement while providing weaker runtime bounds.

\subsection{The Proportional Component}
The component used to compute the proportional MIS, referred to from here on as \Yuvals{}, follows closely the design from~\cite{Emek2013}.
The goal of the rest of the section is to introduce a slightly modified version of their algorithm that allows the non-affected nodes to ignore affected nodes while constructing the initial MIS.

\Yuvals{} consists of states $Q = \{S, D_1, D_2, U_0, U_1, U_2, W, L \}$, where $Q_a = Q - \{W, L \}$ are referred to as \textit{active} states and $q_{yes} = W$ and $q_{no} = L$ as \textit{passive} states.
We set $S$ as the initial state. 
The communication alphabet is identical to the set of states; the algorithm is designed so that node $v$ transmits letter $q$ in round $r$ whenever it resides in state $q$ in round $r + 1$, i.e., state $q$ was returned by the application of the transition function in round $r$. 

\begin{figure*}
	\centering
	\resizebox{12cm}{!}{\includegraphics{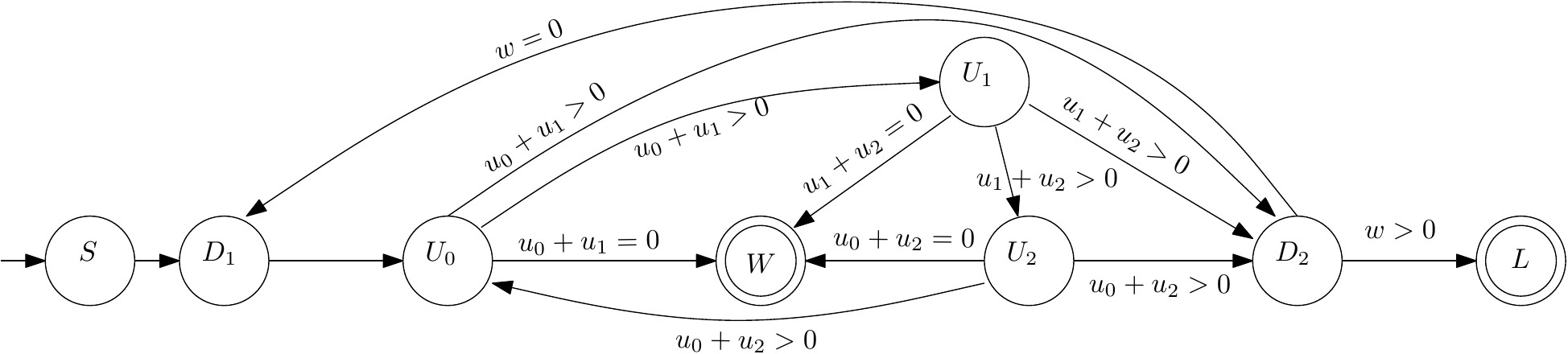}}
	\caption{The transition function of \Yuvals{}. Each edge describes a possible transition and is labeled by the condition that is associated with the corresponding state transition. 
	The lower case letters correspond to the number (counted up to the bounding parameter one) of appearances of the corresponding letter, e.g., the condition $u_1 + u_2 > 0$ is satisfied for node $v$ if at least one port of $v$ contains either letter $U_1$ or $U_2$.
	Each transition $q \rightarrow q'$ is delayed by state $q''$ if there exists a transition $q'' \to q$ (rules associated with delays are omitted from the figure for clarity).}
	\label{figure: fair}
\end{figure*}

\begin{IntuitionSpotlight}
The main idea of \Yuvals{} is that each node $v$ iteratively competes against its neighbors and the winner of a competition enters the MIS.
In every competition, node $v$ tosses fair coins in a synchronized way with its active neighbors.
If a coin toss shows tails, node $v$ \emph{loses} the current competition.
Conversely, node $v$ wins the competition if it is the last node that remains in the competition.
In the case of a tie, a new competition is started.
\end{IntuitionSpotlight}

Let us denote a state transition between states $q$ and $q'$ by $q \rightarrow q'$ omitting the rules associated with this transition.
To ease our notation and to make our illustrations more readable, we say that each state transition $q \rightarrow q'$ is \emph{delayed} by a set $\mathcal{D} = \mathcal{D}(q \rightarrow q') \subseteq Q$ of delaying states. 
For $q \rightarrow q'$ the set of delaying states corresponds to the states in $Q - \{ q \}$ from which there is a state transition to $q$. 
Transition $q \rightarrow q'$ being delayed by state $q''$ indicates that node $v$ does not execute transition $q \rightarrow q'$ as long as there is at least one letter in its ports that corresponds to the state $q''$.
We say that node $v$ in state $q$ is delayed if there is a neighbor $w$ of $v$ that resides in a state that delays at least one of the transitions from $q$, i.e., node $v$ cannot execute some transition because of node $w$.
The transition function of \Yuvals{} is depicted in Figure~\ref{figure: fair}. 

Every competition, for an active node $v$ and its active neighbors, begins by the nodes entering state $U_0$.
During each round every node in state $U_j, j \in \{0, 1, 2\}$, assuming that it is not delayed, tosses a fair coin and proceeds to $U_{j + 1 \mod 3}$ if the coin shows heads and to $D_2$ otherwise. 
Notice that a competition does not necessarily start in the same round for two neighboring nodes.
The circular delaying logic of the $U$-states ensures that node $v$ is always at most one coin toss ahead of its neighbors.
If node $v$ observes that it is the only node in its neighborhood in a $U$-state, it enters state $W$, which corresponds to joining the MIS. 
In the case where, due to unfortunate coin tosses, $v$ moves to state $D_2$ along with all of its neighbors, node $v$ restarts the process by entering state $D_1$.

We call a maximal contiguous sequence of rounds $v$ spends in state $q \in Q_a$ a \textit{$q$-turn} and a maximal contiguous sequence of turns starting from a $D_1$-turn and not including another $D_1$-turn a \textit{tournament}. 
For readability, we sometimes omit the index of a $U_j$-turn for some $j$ and simply write $U$-turn.
We index the tournaments and the turns within the tournament by positive integers.
The correctness of \Yuvals{} in the case of no topology changes follows the same line of arguments as in~\cite{Emek2013}.

We note that dynamically adding edges between nodes in states $U_0, U_1,$ and $U_2$ could potentially cause a deadlock. 
To avoid this, we add a state transition, that is not delayed by any state, from each state $Q - \{ S, L \}$ to state $D'$ (that is part of the fixing component explained in \Sect{}~\ref{sec: fixing}) under the condition that a node reads the initial letter $S$ in at least one of its ports.
In other words, node $v$ enters state $D'$ if an incident edge is added after round $r$ in which $v$ was added into the graph.
Refer to Figure~\ref{figure: complete} for a detailed, though somewhat cluttered, illustration.
\begin{figure}[!h]
 \centering
 \includegraphics[scale=0.6]{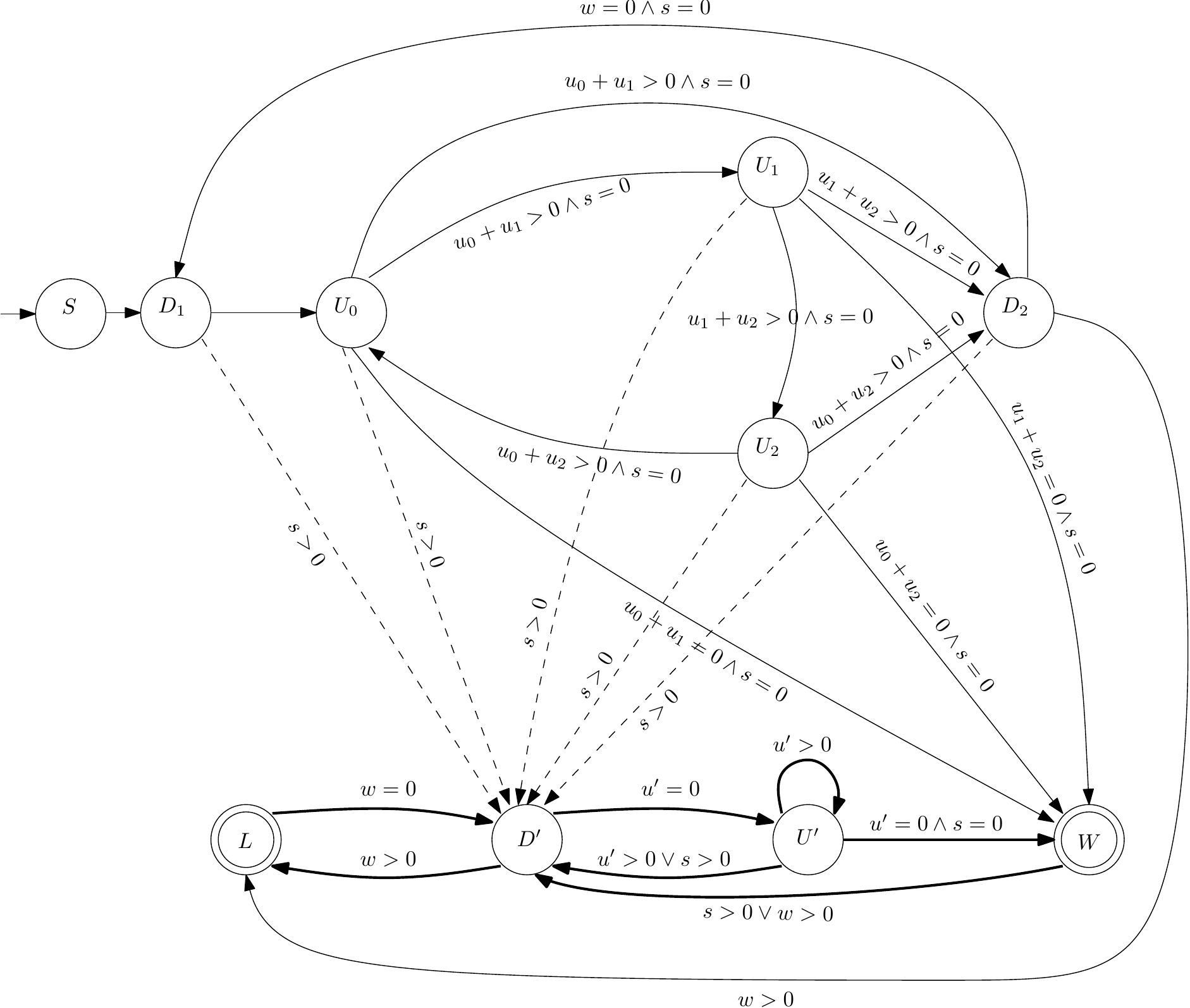}
 \caption{The transition function of our MIS algorithm.
	The dashed lines correspond to the transitions that exclude nodes from the execution of \Yuvals{}.
	A state transition $q \rightarrow q'$ is delayed by state $q''$ if there is a transition indicated by a dashed edge or a thin edge from state $q''$ to $q$. 
	A bold edge from $q$ to $q'$ implies that $q' \rightarrow q''$ is not delayed by $q$.
	As an example, the state transition from $D'$ to $U'$ is not delayed by state $L$ and conversely, the state transition from $D'$ to $U'$ is delayed by state $D_1$.
	Similarly to the other state machine illustrations, the lower case letters in the transition rules correspond to the number of appearances of the corresponding letter and the rules associated with the delays are omitted for clarity. 
	If a node $v$ is in state $q$ and none of the conditions associated with the transitions from $q$ are satisfied (e.g., $v$ is in state $D'$, has a neighbor in state $U'$, and no neighbors in state $W$), then $v$ remains in state $q$. 
	These self-loops are omitted from the picture for clarity.
	}
 \label{figure: complete}
\end{figure}
If node $v$ transitions to state $D'$ due to this condition being met or if $v$ is deleted while being in an active state, we say that $v$ is \emph{excluded} from the execution of \Yuvals{}.

\subsection{The Greedy Component}
\label{sec: fixing}
Now we extend \Yuvals{} to fix the MIS in the case that a topology change leaves the network in an illegal state. 
This extension of \Yuvals{} is referred to as the \emph{greedy component} and denoted by \greedy{}. 
\begin{IntuitionSpotlight}
The basic idea is that nodes in states $L$ and $W$ verify locally that the configuration corresponds to an MIS.
If a node detects that the local configuration does not correspond to an MIS anymore, it revokes its output and tries again to join the MIS according to a slightly modified competition logic.
The crucial difference is that we design Greedy without the circular delay logic. 
The reason behind the design is twofold: First, we cannot afford long
chains of nodes being delayed. Second, since a dynamic change is only allowed to affect its 1-hop
neighborhood, we cannot maintain the local synchrony of adjacent nodes similarly to Proportional.
\end{IntuitionSpotlight}

More precisely, the state set of the algorithm is extended by $Q_2 = \{ U', D' \}$, where states $U'$ or $D'$ are referred to as active. 
The state transition function of the algorithm is illustrated in Figure~\ref{fig: greedy}. 
To detect an invalid configuration, we add the following state transitions from the output states $W$ and $L$:
a transition from state $L$ to state $D'$ in case that a node $v$ resides in state $L$ and does not have any letters $W$ in its ports and a transition from state $W$ to state $D'$ in case $v$ resides in state $W$ and reads a letter $W$ in its ports.
Finally, to prevent nodes in the active states of \Yuvals{} and \greedy{} from entering state $W$ at the same time and thus, inducing an incorrect output configuration, we add a transition from $W$ and $U'$ to $D'$ in case that node $v$ reads the initial letter $S$.

The logic of the new states $U'$ and $D'$ is the following: node $v$ in state $D'$ goes into state $U'$ if it does not read the letter $U'$ in its ports. 
Then $u$ and its neighbors that transitioned from $D'$ to $U'$ during the same round compete against each other:
In every round, $v$ tosses a fair coin and goes back to state $D'$ if the coin shows tails. We emphasize that nodes in state $U'$ are not delayed by nodes in state $D'$. 
Then, if $v$ remains the only node in its neighborhood in state $U'$, it declares itself as a winner and moves to state $W$. 
Conversely, if a neighbor of $v$ wins, $v$ goes into state $L$.

\begin{figure}
	\centering
	\resizebox{12cm}{!}{\includegraphics{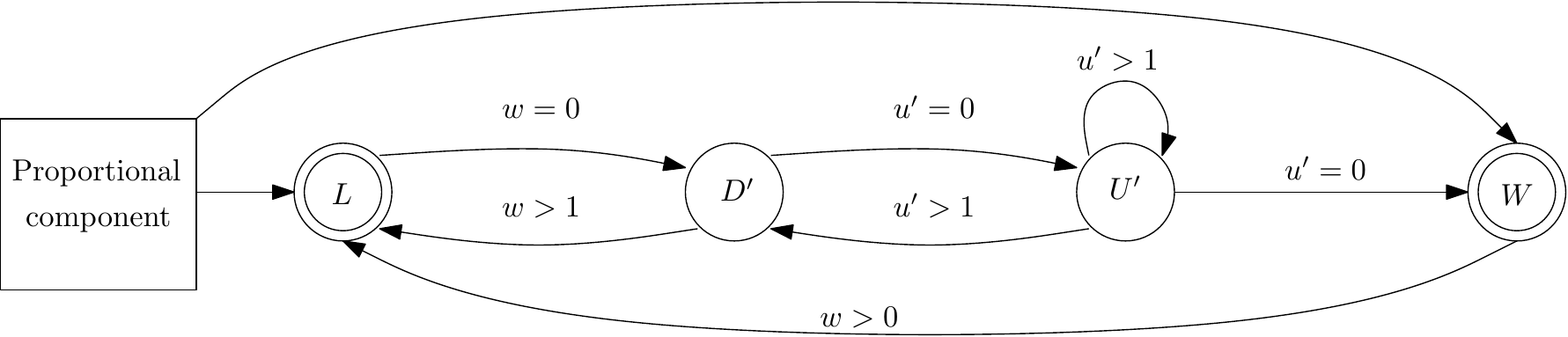}}
	\caption{The transition function of \greedy{}. The conditions for the transitions are denoted by the labels on the edges and a lower case letter corresponds to the number of appearances of the corresponding letter. Unlike in \Yuvals{}, a transition from $q'$ to $q$ does not imply that $q'$ delays any transition from $q$.}
	\label{fig: greedy}
\end{figure}

\begin{observation}
	A non-affected node is never in an active state of \greedy{}.
	\label{obs: non-affectedNotGreedy}
\end{observation}
\begin{proof}
	Consider a non-affected node $u$.
Since $u$ is non-affected, it will never read the initial letter $S$ in its ports after the first round. 
	The design of \Yuvals{} guarantees that either $u$ enters state $W$ and no neighbor of $u$ enters state $W$ or $u$ enters state $L$ and at least one neighbor $v$ of $u$ enters state $W$. 
	Due to the design of our algorithm, a node will never enter state $W$ if it has a neighbor in state $W$ and therefore, if $u$ once enters state $W$, it will never reside in any other state.
	By the same argument, if $u$ is in state $L$, the neighbor $v$ in state $W$ will guarantee that $u$ never leaves state $L$.
\end{proof}
\section{Analysis}
\label{sec: analysis}
\paragraph{An Auxiliary Execution of \Yuvals{}}
The authors of \cite{Emek2013} analyze an algorithm very similar to the one
induced by \Yuvals{}, showing that its runtime (on a static graph) is $O
(\log^{2} n)$ in expectation and w.h.p.
In their analysis, they prove that with a positive constant probability,
the number of edges decreases exponentially from one tournament to the next.
\begin{IntuitionSpotlight}
The dynamic nature of the graph $\mathcal{G}$ prevents us from analyzing the
execution of \Yuvals{} in the same manner.
To overcome this obstacle, we introduce an auxiliary execution $\pkaks$ of
\Yuvals{} that turns out to be much easier to analyze, yet serves as a good
approximation for the real execution.
\end{IntuitionSpotlight}

Consider some execution $\eta$ of \Yuvals{} on $\mathcal{G}$ and recall that
$\eta$ is determined by the adversarial graph sequence $\mathcal{G}$ and by
the coin tosses of the nodes.
Let the \emph{last} $U$-turn stand for a shorthand notation for a $U_j$-turn, where $j$ is the largest index for which a $U_j$ turn exists in tournament $i$, under $\pkaks(\eta)$.
We associate with $\eta$ the \emph{auxiliary execution}
$\pkaks = \pkaks(\eta)$
obtained from $\eta$ by copying the same coin tosses (for each node) and
modifying the adversarial policy by applying to it the following two
rules for
$i = 1, 2, \dots$:
(i)
If node $u$ is \emph{excluded} during tournament $i$ under $\eta$, then,
instead, we \emph{delete} node $u$ immediately following its last $U$-turn.
(ii)
If node $u$ becomes passive during tournament $i$ under $\eta$ and it would
have remained active throughout tournament $i$ under $\pkaks(\eta)$, then we
also \emph{delete} node $u$ immediately following its last $U$-turn
under $\pkaks(\eta)$.

Let $V_{\eta}(i)$ and $V(i)$ be the sets of nodes for which tournament $i$
exists under $\eta$ and $\pkaks(\eta)$, respectively.
Let $X^{v}(i)$ and $Y^{v}(i)$ denote the number of $U$-turns of node $v$ in
tournament $i$ under $\eta$ and $\pkaks(\eta)$, respectively;
to ensure that $X^{v}(i)$ and $Y^{v}(i)$ are well defined, we set
$X^{v}(i) = 0$
if
$v \notin V_{\eta}(i)$
and
$Y^{v}(i) = 0$
if
$v \notin V(i)$,
i.e.,
if tournament $i$ does not exist for $v$ under $\eta$ and $\pkaks(\eta)$,
respectively.
Let $\nbh_{\eta}(v, i)$ and $\nbh(v, i)$ denote the (exclusive) neighborhood
of $v$ at the beginning of tournament $i$ under $\eta$ and $\pkaks(\eta)$,
respectively.
Observe that since state $D_2$ is delayed by states $U_j$ for all $j$, no neighbor of node $u$ can enter tournament $i + 1$ while node $u$ is still in tournament $i$. 
Therefore,
$\nbh_{\eta}(v, i) \subseteq V_{\eta}(i)$
and
$\nbh(v, i) \subseteq V(i)$.
We say that node $v$ \emph{wins} in tournament $i$ under $\eta$ if
$X^{v}(i) > X^{w}(i)$
for all
$w \in \nbh_{\eta}(v, i)$
and if $v$ is not excluded in tournament $i$;
likewise, we say that node $v$ \emph{wins} in tournament $i$ under
$\pkaks(\eta)$ if
$Y^{v}(i) > Y^{w}(i)$
for all
$w \in \nbh(v, i)$
and if $v$ is not deleted in tournament $i$.
The following lemma plays a key role in showing that the number of tournaments
in $\pkaks(\eta)$ is at least as large as that in $\eta$.

\begin{lemma}
If
$v \in V(i)$
wins in tournament $i$ under $\pkaks(\eta)$, then $v$ wins in tournament $i$
under $\eta$.
\label{lemma:  P2>fair}
\end{lemma}
\begin{proof}
	First, we show that $X^w(i) \leq Y^w(i)$ for any tournament $i$ that exists in both $\eta$ and $\pkaks(\eta)$ for any node $w$. 	
	We observe that if $w$ is not excluded in tournament $i$, the number of $U$-turns is uniquely determined by the coin tosses in both executions, i.e., $X^w(i) = Y^w(i)$. 
	Then, if $w$ is excluded in tournament $i$, there are at least as many $U$-turns in $\pkaks(\eta)$ as in $\eta$. 
	Thus, $X^w(i) \leq Y^w(i)$. 

	Assume then that some node $v$ wins in tournament $i$ according to $\pkaks(\eta)$. 
	Since $v$ can only win in tournament $i$ if it is not excluded in this tournament, we get that $X^v(i) = Y^v(i)$. 
	Node $v$ winning in tournament $i$ implies $X^v(i) = Y^v(i) > Y^w(i) \geq X^w(i)$ for any $w \in \nbh(v, i)$ and thus, $v$ also wins in $\eta$.
	
	To conclude the proof, we observe that if tournament $i$ does not exist for node $v$ in $\pkaks(\eta)$, then either $v$ was deleted before entering state $D_1$ of tournament $i$ or either $v$ or some of its neighbors won in tournament $j < i$.
	In all aforementioned cases, tournament $i$ also does not exist for $v$ in $\eta$.
	Finally, by construction, if tournament $i$ does not exist for $v$ in $\eta$, it also does not exist for $v$ in $\pkaks(\eta)$ and the claim follows.
\end{proof}

\subsection{Runtime of \Yuvals{}}
\label{sec: analysisProportional}
Next, we analyze the number of tournaments in $\pkaks(\eta)$. 
Once we have a bound on the number of tournaments executed in $\pkaks(\eta)$, it is fairly easy to obtain the same bound for $\eta$. 
Similarly as before, the set of nodes for which tournament $i$ exists according to $\pkaks(\eta)$ is denoted by $V(i)$. 
Furthermore, let $\nbh(v, i)$ be the set of neighbors of $v$ for which tournament $i$ exists in $\pkaks(\eta)$ and let $E(i) = \bigcup_{v \in V(i)} \{ \{v, w\} \mid w \in \nbh(v, i) \}$ and $G(i) = (V(i), E(i))$. 
Notice that $G(i)$ is defined for the sake of the analysis and the topology of the underlying graph does not necessarily correspond to $G_i$ in any round round $i$, where $G_i$ is the $i$th element in the adversarial graph sequence.

According to the design of \Yuvals{}, no node that once enters a passive state can enter an active state of \Yuvals{}. 
Furthermore, if any edge is added adjacent to node $v$ after $v$ has transitioned out of state $S$, node $v$ will not participate in any tournament after the addition.
Therefore, we get that $V(i + 1) \subseteq V(i)$ and $E(i + 1) \subseteq E(i)$ for any $i$. 
The following lemma plays a crucial role in the runtime analysis of \Yuvals{}. 
\begin{IntuitionSpotlight}
In the case of no topology changes, the proof is (almost) the same as the one
in~\cite{Emek2013}.
By a delicate adjustment in the details of the proof, we show that the same
argument holds in the presence of topology changes and only affects the
constant $\ell$ by a constant factor.
\end{IntuitionSpotlight}

\begin{lemma}
	There are two constant $0 < p, \ell < 1$ such that $|E(i+1)| \leq \ell |E(i)|$ with probability $p$.
	\label{lemma: edgesRemoved}
\end{lemma}
\begin{proof}
Let $d(v, i) = |\nbh(v, i)|$. 
We say that a node $v$ is \textit{good} in $G(i)$ if 
\[
	|\{ w \in \nbh(v, i) \ | \ d(v, i) \geq d(w, i) \}| \geq d(v, i) / 3 \ .
\] 
It is known, that at least half of the edges of any graph are adjacent to good nodes~\cite{ABI86}.	
For the following part of the proof, we use techniques similar to the ones in~\cite{Emek2013}.
The main idea is that a good node is covered by a node with a smaller degree with a constant probability in any tournament and thus, a constant fraction of the edges is removed in each tournament.
The subtle difference is that in the case of a dynamic network, there is no guarantee that a node adjacent to a good node remains adjacent until the end of the tournament.
We tackle this by considering the auxiliary execution $\pkaks(\eta)$ and showing that the adversary essentially has to either delete a constant fraction of the edges in the graph or otherwise, we can use the previous techniques.

Let us now consider some good node $v \in G(i)$.	
Our goal is to show that at least one sixth of the edges connected to $v$ in tournament $i$ do not exist in $V(i + 1)$ with a constant probability. 
Recall that we are now considering the auxiliary execution $\pkaks(\eta)$ in which nodes are deleted in the end of a tournament if they are excluded in this tournament according to $\eta$.
We split our analysis into two cases, where the first case considers the option that in tournament $i$, the adversary either causes the deletion of $v$ or of at least half of the edges between $v$ and the nodes in $\{ w \in \nbh(v, i) \ | \ d(v, i) \geq d(w, i) \}$.
Notice that the adversary can do this by either directly deleting the edges, by deleting nodes in $\{ w \in \nbh(v, i) \ | \ d(v, i) \geq d(w, i) \}$ or by adding edges adjacent to nodes in these sets, which results in the deletion of the corresponding nodes in the auxiliary execution.
It follows that at least one sixth of edges connected to $v$ do not exist in $V(i + 1)$ with probability $1$.
	
Let us then assume for the second case that the adversary causes the deletion at most half of the edges connected to 
\[
	\{ w \in \nbh(v, i) \ | \ d(v, i) \geq d(w, i) \}
\]
in tournament $i$. 
Let $A(w, i)$ be the event that $w$ is not deleted in tournament $i$ and let
	\[
		\hnbh(v, i) =  \{ w \in \nbh(v, i) \ | \ A(w, i) \land \left( d(w, i) \geq d(v, i) \right) \}  \ .
	\] 
We say that node $w \in \hnbh(v, i)$ wins $v$ in tournament $i$ if
	\[
		X^{w}(i) > \max\{ X^{w}(i)\ |\ w \in \nbh(w, i) \cup \hnbh(v, i) - \{ w \} \}
	\]
	and denote this event by $\WIN_i(w, v)$. The idea is that if $w$ wins $v$ in tournament $i$, then $w$ enters the W state and $v$ enters the L state. Furthermore, the events $\WIN_i (w, v)$ and $\WIN_i (z, v)$ are disjoint for every $w, z \in \hnbh(v, i), w \neq z$.

	Let us now fix a node $w \in \hnbh(v, i)$. We denote the event that the maximum of $\{ X^{z} |\ z \in \nbh(w, i) \cup \hnbh(v, i) \}$ is attained at a single $z \in \nbh(w, i) \cup \hnbh(v, i)$ by $B_i (w, v)$. Since 
\[
	| \nbh(w, i) \cup \hnbh(v, i) | \leq 2d(v, i)
\]
by the definition of a good node and $X^{z} (i)$ are independent geometric random variables for all $z \in \nbh(w, i) \cup \hnbh(v, i)$, we get that
	\[
		\Probability( \WIN_i (w, v) ) = \Probability ( \WIN_i( w, v)\ |\ B_i (w, v)) \cdot \Probability( B_i (w, v)) \geq \frac{1}{2d(v, i)} \cdot \frac{2}{3} \ .
	\]
	
	Since $v$ is good in $G(i)$ and the events $\WIN_i (w, v)$ disjoint, we get that
	\[
		\begin{split}
			\Probability \left( v \notin V(i + 1) \right) 	& \geq \Probability \left( \bigvee_{w \in \hnbh(v, i)} \WIN_{i}(w, v) \right) \\
															& = \sum_{w \in \hnbh(v, i)} \Probability \left( \WIN_{i}(w, v) \right) \\
															& \geq \frac{d(v, i)}{6} \cdot \frac{1}{2d(v, i)} \cdot \frac{2}{3} = \frac{1}{18} \, .
		\end{split}
	\]
	Recalling that half of the edges are connected to good nodes, we get that 
\[
	\Expectation[|E(i + 1)|] < \frac{35|E(i)|}{36}  \ .
\]
The claim now follows by Markov's inequality.	
\end{proof}

\begin{theorem}
	Any node participates in $\bigO(\log n)$ tournaments before becoming passive \whp and in expectation.
	\label{thm: NrTournaments}
\end{theorem}
\begin{proof}
	Let $Z = \min\{ 0 \leq i \in \mathbb{Z} \ | \ |E(i)| = 0 \}$. 
By Lemma~\ref{lemma: edgesRemoved}, $Z$ is dominated by a random variable that follows distribution 
\[
	\bigO(\log n) + \textrm{NB}(\bigO(\log n), 1 - p) \ ,
\] 
where \textrm{NB} stand for the negative binomial distribution with failure probability $1 - (1 - p) = p$ and $\bigO(\log n)$ as the number of failed trials until the process is stopped.
Put otherwise, $\textrm{NB}(\bigO(\log n), 1 - p)$ bounds the number of tournaments needed until the number of edges in our input graph is reduced by a constant factor $\bigO(\log n)$ times.
Therefore, $Z = \bigO(\log n)$ in expectation and \whp{}
	
	According to the design of the algorithm, a node that is not excluded and does not have any active neighbors in tournament $i$ goes into state $W$ with probability $1$. 
Therefore, all active nodes in tournament with index $Z$ will either become passive or be deleted. 
Then, by applying Lemma~\ref{lemma:  P2>fair} and observing that no new edges can be added between nodes in active states of \Yuvals{}, we get that there are no nodes for which tournament $Z + 1$ exists in $\eta$.
\end{proof}

The last step of the analysis of \Yuvals{} is to bound the number of rounds that any node $v$ spends in an active state. 
Let $V_{\textrm{MAX}}(v)$ denote the maximal connected component of nodes in active states of \Yuvals{} that contains node $v$.
Consider the following modification of \Yuvals{}: before starting tournament $i + 1$, node $v$ waits for every other node in $V_\textrm{MAX}(v)$ to finish tournament $i$ or to become excluded. 
Notice that any node that gets inserted into the graph after node $v$ has exited state $S$ will not be part of $V_\textrm{MAX}(v)$.
We do not claim that we know how to implement such a modification, but clearly the modified process is not faster than the original one. 

The length of tournament $i$ with respect to the above modification is determined by $\max_v \{ X^v(i) \}$. Given that the random variables $X^v(i)$ are independent and follow the $\textrm{Geom}(1/2)$ distribution, we get that $\max_v \{ X^v(i) \} \in \bigO(\log n)$ with high probability and in expectation.
Combining with Observation~\ref{obs: non-affectedNotGreedy}, we get the following corollaries.
\begin{corollary}
	Let $t_v \geq 1$ be the round in which node $v$ is inserted into the graph, where $t_v = 1$ if $v \in G_1$.
	If $v$ is not deleted, it enters either state $W$ or $L$ by time  $t_v + \bigO(\log^2 n)$ \whp and in expectation.
	\label{cor: FairRunTime}
\end{corollary}
\begin{corollary}
	The maximum runtime of any non-affected node is $\bigO(\log^2 n)$ in expectation and \whp{}
	\label{cor: non-affected}
\end{corollary}

\subsection{The Quality of an MIS}
\label{sec: quality}
Before we go to the analysis of \greedy{}, we want to point out that even though an invalid configuration can be detected locally, a single topology change can have an influence in an arbitrary subgraph.
\begin{IntuitionSpotlight}
To demonstrate this observation, consider a graph $G$ with some node $v$ that
is connected to every other node in some subgraph $G'$ of $G$ and assume, that after executing
\Yuvals{}, the MIS contains node $v$.
Given that the adversary deletes node $v$, the MIS has to be fixed on $G'$, that can have an arbitrary topology, e.g., a large diameter.
Furthermore, several deletions like this can potentially happen one after another in a set of subgraphs of $G$ that are connected,
which can result in some neighboring nodes to start the execution of the fixing component at different times.
\end{IntuitionSpotlight}

We begin the analysis of \greedy{} by taking a closer look at the properties of non-affected nodes in the passive states and show that if a node has a high degree, it is either unlikely for this node to be in the MIS or that the neighbors of this node have more than one MIS node in their neighborhoods.
Let $i$ be the index of the tournament in which node $v$ enters state $W$. 
We say that node $v$ \textit{covers} node $w \in \nbh(v, i) \cup \{ v \}$ if $w$ entered state $L$ in tournament $i$.  
\begin{definition}
	The quality $q(v)$ of node $v$ is given by $q(v) = |\{ w \in \nbh(v, i) \cup \{ v \} \mid v \text{ covers } w \}|$.
	The quality $q(B)$ of any set of nodes $B$ is defined as $\sum_{v \in B} q(v)$.
\end{definition}


\begin{lemma}
	For any node $v$, $\Expectation[q(v)] \in \bigO(\log n)$.
	\label{lemma: quality}
\end{lemma}
\begin{proof}
	Consider node $v$ and tournament $i$. 
Let $\hnbh(v, i) \subseteq \nbh(v, i)$ be the neighbors of $v$ that are not excluded in tournament $i$. 
We note that $v$ can only cover $w \in \nbh(v, i)$ if $v$ is not excluded in tournament $i$. 
Let $q(v, i)$ denote the random variable that counts the number of nodes that $v$ covers in tournament $i$. 
Since $X^v(i)$ and $X^w(i)$ for all $w \in \nbh(v, i)$ are independent $\textrm{Geom}(1/2)$ random variables, we get that
\[
	\begin{split}
		\Expectation[q(v, i)] 	& \leq \Probability\left[X^v(i) \geq \max \{ X^w(i) \ | \ w \in \hat \nbh(v, i) \} \right] \cdot (|\hnbh_i(v)| + 1) \\
								& \leq 2 \cdot \Probability\left[X^v(i) > \max \{ X^w(i) \ | \ w \in \hat \nbh(v, i) \} \right] \cdot (|\hnbh_i(v)| + 1) \\
								& \leq \frac{2(|\hnbh(v, i)| + 1)}{|\hnbh(v, i)| + 1}  = 2 \ . 
	\end{split} 
\]

Consider the random variable $Z = \min \{ j \in \mathbb{Z}_{\geq 0} \ | \ |E(j)| = 0 \}$.
The total number of tournaments is at most $Z$ and $\Expectation[q(v, i)] = \Expectation[q(v, i) \ | \ Z \geq i] \cdot \Probability[Z \geq i]$. 
Thus, by Lemma~\ref{lemma: edgesRemoved}
\[
	\begin{split}
		\Expectation[q(v)] & = \sum_{i = 0}^{\infty} \Expectation[q(v, i) \ | \ Z \geq i] \cdot \Probability[Z \geq i] \in \bigO(1) \cdot \left( \sum_{i = 0}^{\infty} \Probability[Z \geq i] \right) \subseteq \bigO(\log n) \ . \hfill 
	\end{split}  
\] 

\end{proof}


The quality of a node $v$ gives an upper bound on the number of nodes in the neighborhood of $v$ that are only covered by $v$.
Let $v$ be a node in state $L$. If $t$ is the first round such that there are no nodes adjacent to $v$ in state $W$, we say that $v$ is \textit{released} at time $t$. 
Similarly, we say that node $v$ in state $W$ is released if an adjacent edge is added.
Every node can only be released once, i.e., node $v$ counts as released even if it eventually again has a neighboring node in state $W$ or if it enters state $W$. 
\begin{lemma}
	Let $H$ be the set of nodes that are eventually released. Then $\Expectation[|H|] \in \bigO(C \log n)$.
	\label{lemma: participate}
\end{lemma}
\begin{proof}
	Let $V^c$ be the set of nodes that are (eventually) deleted or incident to an edge which is (eventually) added or deleted. 
By Lemma~\ref{lemma: quality}, the quality of any node is $\bigO(\log n)$. 
Therefore, $\Expectation[q(V^c)] = \sum_{v \in V^c} \Expectation[q(v)] \in |V^c| \cdot \bigO(\log n)$. 
Thus, $\Expectation[|H|] \in |V^c| \cdot \bigO(\log n) \subseteq \bigO(C \log n)$.
\end{proof}

\subsection{Fixing the MIS}
\label{sec: analysisGreedy}
We call a maximal contiguous sequence of rounds in which node $v$ resides in state $U'$ a \textit{greedy tournament}. 
Unlike with \Yuvals{}, we index these greedy tournaments by the time this particular tournament starts. 
In other words, if node $v$ resides in state $D'$ in round $t - 1$ and in state $U'$ in round $t$, we index this tournament by $t$.
We denote the random variable that counts the number of transitions from $U'$ to itself in greedy tournament $t$ by node $v$ by $X^v (t)$. 
Notice that $X^v (t)$ obeys $\textrm{Geom}(1/2)$ unless $v$ is deleted or an adjacent edge is added during greedy tournament $t$.
We say that a greedy tournament $t$ is \textit{active} if there is at least one node $v$ in state $U'$ of greedy tournament $t$.

\begin{observation}
	There are at most $\bigO(\log n)$ rounds in which greedy tournament $t$ is active in expectation and \whp
	\label{obs: lenGreedy}
\end{observation}
\begin{proof}
Let $V(t)$ be the set of nodes that participate in greedy tournament $t$. 
%
Recall that the number of state transitions in greedy tournament is exactly $X^v (t)$. 
Furthermore, $X^v (t)$ follows distribution \textrm{Geom}$(1/2)$. 
Since $|V(t)| \leq n$, all $X^w (t), w \in V(t)$ are independent and the maximum of $X^w (t)$ for at most $n$ variables is $\bigO(\log n)$ \whp and in expectation, the claim follows.
\end{proof}


\begin{IntuitionSpotlight}
In every greedy tournament $i$, either a node joins the MIS with a constant probability or a topology change occurs that affects some node.
Furthermore, any single topology change can force at most two nodes to leave state $W$.
Therefore, the sum of transitions to state $W$ by the nodes that ever execute \greedy{}, i.e., nodes that are either released or excluded, can be at most $\bigO( H + C )$.
Thus, there can be at most $\bigO( H + C )$ greedy tournaments in total.
\end{IntuitionSpotlight}

\begin{lemma}
	The total expected number of greedy tournaments is $\bigO(C \log n)$.
	\label{lemma: NumberGreedy}
\end{lemma}
\begin{proof}
To prove the claim, we bound from above the sum of state transitions into state $W$ by the released nodes and the excluded nodes. 
A node transitions away from state $W$ at most once per added incident edge and therefore, there can be at most $2C$ transitions from $W$ to $D'$ in total.
After at most $2C + |H|$ transitions into state $W$ by the released nodes, all released nodes are either in state $W$ or deleted.


We say that a greedy tournament $t$ is \emph{unclean} if either, some node $v$ in an active state of greedy tournament $t$ is deleted or an edge is added adjacent to $v$.
Otherwise, we say that greedy tournament $t$ is \textit{clean}.
Let $V(t)$ be the nodes that participate in greedy tournament $t$, i.e., the nodes that reside in state $U'$ in round $t$ and in state $D'$ in round $t-1$. 
Assuming that greedy tournament $t$ is clean and recalling that the random variables $X^v (t)$ are independent, the probability that the maximum of $X^v (t), v \in V(t)$ is attained in a single node is at least $1/3$. 
Let $Y'_k$ be the random variable that counts the number of clean greedy tournaments we need until the sum of state transitions to $W$ by the released nodes is $k$. 
It is easy to see that $Y'_k$ is dominated by random variable $Y_k$ that obeys distribution $k + \textrm{NB(k, 2/3)}$. 


Now we set $k = 2C + |H|$ and get that $\Expectation[Y_k] \in \bigO(k) \subseteq \bigO(C + |H|)$. 
By Lemma~\ref{lemma: participate}, $|H| \in \bigO(C \log n)$ in expectation. 
Thus, we get that 
\[
	\Expectation[Y'_{k}] \in \bigO(C + |H|) \subseteq \bigO(C \log n) \ .
\] 
Finally, we observe that there can be at most $2C$ unclean greedy tournaments and thus, the total expected number of greedy tournaments is $\bigO(C \log n) + \bigO(C) = \bigO(C \log n)$.
\end{proof}
\begin{observation}
	Let $k$ be the total number of greedy tournaments. There are at most $2k + 2C$ non-silent rounds without either $i)$ at least one node in an active state of \Yuvals{} or $ii)$ an active greedy tournament.
	\label{obs: uselessNonSilent}
\end{observation}
\begin{proof}
	Consider a non-silent round $t$ and assume that in rounds $t$ and $t + 1$ no changes occur, that no greedy tournaments are active, and that no node is in an active state of \Yuvals{}. 
Notice that no node $v$ can be in state $D'$ in round $t$. 
Otherwise, $v$ transitions to state $U'$ in round $t + 1$ and thus, greedy tournament $t + 1$ is active in round $t + 1$.
Then, in graph $G_t$, either $i)$ there exists at least one pair of adjacent nodes $v, w$ that are in state $W$ or $ii)$ there is a node $v'$ that is in state $L$ and none of its neighbors are in state $W$.

In case $ii)$, node $v'$ can have only letters $S$ and $L$ in its ports in round $t$. 
Therefore, the logic of the $L$ state ensures that node $v'$ is in state $D'$ in round $t + 1$.
Now, we get that greedy tournament $t + 2$ is active in round $t + 2$.

In case $i)$, both nodes $v$ and $w$ have at least one letter $S$ or $W$ in their ports in round $t$.
Therefore, the logic of the $W$ state ensures that, both $v$ and $w$ reside in state $D'$ in round $t+1$.
Again we get that a greedy tournament has to be active in round $t + 2$. 
Thus, given no topology changes, there can be at most two successive non-silent rounds without an active greedy tournament or at least one node active in \Yuvals{} and the claim follows.
	
\end{proof}

Next we show that our MIS algorithm, that consists of \Yuvals{} and \greedy{}, fulfills the correctness properties given in the model section. 
Then, we establish our main result, i.e., that our MIS algorithm is indeed effectively confining. 
\begin{lemma}
 Our MIS algorithm is correct. 
 \label{lemma: correctness}
\end{lemma}
\begin{proof}
	The correctness property $(C3)$ follows directly from the logic of the states $W$ and $L$.
	To obtain property $(C1)$, we first observe that if the nodes are in an incorrect output configuration in round $r$, there has to be either some node in state $W$ with another node in state $W$ in its neighborhood or some node along with all of its neighbors in state $L$.
	In both cases, the logic of states $W$ and $L$ ensure that the configuration, given no topology changes, is a non-output configuration in round $r + 1$.
	Given Corollary~\ref{cor: FairRunTime}, we get that eventually, no active nodes in \Yuvals{} remain once there are no more topology changes.
	Combining this with Observation~\ref{obs: uselessNonSilent} and Lemma~\ref{lemma: NumberGreedy}, we get property $(C1)$.
		
	Consider some node $v$ in an active state in round $r$.
	According to the design of our protocol, if an edge is added adjacent to $v$ in round $r$, then $v$ will not reside in state $W$ in round $r + 1$.
	We observe that due to the logic of the excluding transitions, i.e., the transitions to state $D'$ when reading the letter $S$, node $v$ that participates in a tournament of \Yuvals{} cannot have any neighbors that participate in a greedy tournament. 
	Furthermore, node $v$ can only enter state $W$ by winning a tournament (either in \Yuvals{} or in \greedy{}).
	Therefore, it cannot be the case that both $v$ and some neighbor $w$ of $v$ enter state $W$ in round $r + 1$.	
	Furthermore, node $v$ can only enter state $L$ in round $r + 1$ if it has a neighbor in state $W$ in round $r$.
	Thus, if node $v$ is in a non-output state in round $r$, it will either be in a non-output state, enter state $W$ and none of its neighbors will, or enter state $L$ and at least one of its neighbors will enter state $W$ in round $r + 1$.
	
	Consider then the case that node $v$ is in an output state in round $r$.
	Due to the logic of state $W$, node $v$ only remains in state $W$ if there are no letters $S$ or $W$ in its ports.
	Since a port that corresponds to a neighbor in state $W$ can only contain letters $W$ or $S$, $v$ cannot have a neighbor in state $W$ in round $r$.
	Due to the logic of state $L$, $v$ only remains in $L$ if there is at least one neighbor in state $W$.
	Therefore, either the configuration in round $r+1$ is not an output configuration or the configuration is a correct output configuration and we have obtained property $(C2)$.
\end{proof}

\begin{theorem}
	\label{thm: complexity}
	The global runtime of our MIS algorithm is $\bigO( (C + 1) \log^2 n )$ in expectation.
\end{theorem}
\begin{proof}
By Lemma~\ref{lemma: NumberGreedy}, we have $\bigO(C \log n)$ active greedy tournaments in expectation and by Observation~\ref{obs: lenGreedy} each greedy tournament takes $\bigO(\log n)$ rounds \whp and in expectation. 
Thus, in expectation, we have $\bigO(C \log^2 n)$ non-silent rounds where a greedy tournament is active. 
Combining this with Corollary~\ref{cor: FairRunTime} and Observation~\ref{obs: uselessNonSilent} and by linearity of expectation, the total expected runtime is $\bigO(\log^2 n) + \bigO(C \log n) + \bigO(C \log^2 n) \subseteq \bigO((C + 1) \log^2 n)$.
\end{proof}

\section{Lower Bound}
\label{sec: lower}
The runtime of our algorithm might seem rather slow at the first glance, since it is linear in the number of topology changes. 
In this section, we show that one cannot get rid of the linear dependency, i.e., there are graphs where the runtime of any algorithm grows at least linearly with the number of changes. 
In particular, we construct a graph where the runtime of any algorithm is $\Omega(C)$.

\begin{IntuitionSpotlight}
Even in a graph with two connected nodes, any algorithm has to perform some sort of symmetry breaking. Furthermore, if the degrees of node $u$ and its neighbors are small, it is likely that deleting $u$ leaves some of its neigbhors in an invalid state.

Let $G^\ell$ be a graph that consists of $n$ nodes and $n = 3\ell$. The graph consists of $\ell$ components $B_i$, where $B_i$ is a $3$-clique for each $1 \leq i \leq \ell$. In addition, let $u_1, \ldots, u_\ell$ be a set of nodes such that $u_i \in B_i$ for every $i$. For illustration, see Figure~\ref{fig: lower}.
\end{IntuitionSpotlight}

\begin{figure}
	\centering
	\resizebox{10cm}{!}{\includegraphics{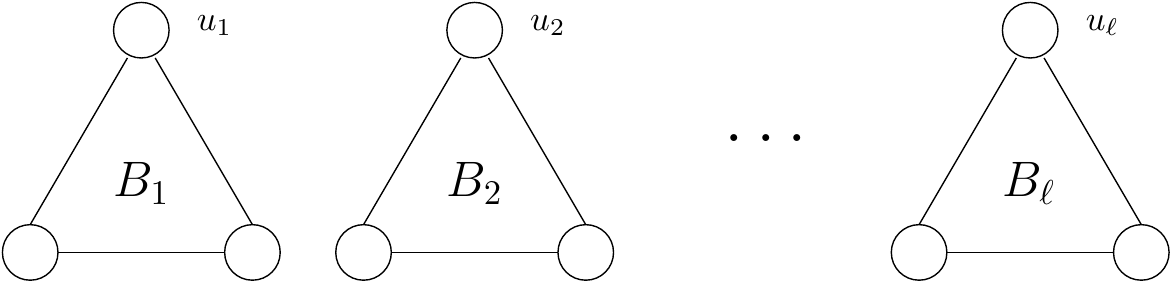}}
	\caption{One node in each of the connected components has to be in the MIS. Since the nodes in the components are indistinguishable and execute the same algorithm, the probability of $u_i$ being in the MIS for any $i$ is $1/3$.}
	\label{fig: lower}
\end{figure}

\begin{theorem}
	\label{thm: lower}
	There exists a graph $G$ and a schedule of updates such that the runtime of any algorithm on $G$ is $\Omega(C)$ in expectation and \whp
\end{theorem}
\begin{proof}
	Consider graph $G^\ell$ and the following adversarial strategy. 
In round $2i - 1 \geq 1$, the adversary deletes node $u_{i}$, i.e., one of the nodes in component $B_{i}$. 
Our goal is to show that at least a constant fraction of the first $2C$ rounds are non-silent for any $C \leq n/3$. 
Consider round $2i$ and component $B_{i}$. 
Since the nodes in component $B_{i}$ form a clique, their views are identical. 
Therefore, according to any algorithm that computes an MIS, their probability to join the MIS is equal, i.e., $1/3$. 
Thus, the probability that the nodes in $B_{i} - \{u_{i}\}$ are not in the MIS in any round $2j \leq 2i$ is at least $1/3$.

Let $X_i$ then be the indicator random variable, where $X_i = 0$ if round $i$ is silent and $1$ otherwise and let $X = \sum_{i}^{\infty} X_i$ be the random variable that
counts the number of non-silent rounds. 
Since any round $1 \leq 2i \leq 2C$ is non-silent with at least probability $1/3$, we get that $\Expectation[X] \geq (1/3) C$. 
Now by applying a Chernoff bound, we get that 
\[
	\Probability\left[ X < 1/2 \Expectation[X]  \right] = \Probability\left[ X < (1/2) \cdot (1/3) C  \right] \leq  2^{-C/12}  \ .
\]
Since $C = n/3$, we get that $\Probability\left[ X < 1/2 \Expectation[X]  \right] \in \bigO(n^{-k})$ for any constant $k$ and thus, the claim follows.
\end{proof}

\section{Pseudo-Locality}
\label{sec: logDist}
In this section, we show that the runtime of node $u$ is independent of topology updates further than $1 + \log n$ distance from $u$. 
We say that an update occurs in distance $d$ from node $u$ if it involves a node within the $d$-hop neighborhood of $u$, where involving node $u$ indicates that either $u$ is inserted/deleted or an edge adjacent to $u$ is inserted/deleted.
Let $N^i_r(u)$ denote the (inclusive) $i$-hop neighborhood of node $u$ in round $r$.
Then, we denote the number of topology changes that involve the nodes in 
\[
	\bigcup_{r = 0}^{\infty} N^{1 + \log n}_r(u) = N^{1 + \log n}_\infty (u)
\]
by $C_u$.
In other words, $C_u$ counts the number of changes that involve nodes that have been within $1 + \log n$ hops from $u$ in some round.

We start by using Lemma~\ref{lemma: quality} to bound the expected number of released nodes in $N^{\log n}_\infty(u)$. 
The quality of the nodes in $N^{1 + \log n}_\infty(u)$ gives an upper bound to the number of released nodes in $N^{\log n}_\infty(u)$. 
The proof for the following lemma is analogous to the proof of Lemma~\ref{lemma: participate} and therefore omitted.

\begin{lemma}
	Let $H_u$ be the set of released nodes in $N^{\log n}_\infty(u)$. Then $\Expectation[|H_u|] \in \bigO(C_u \log n)$.
	\label{lemma: pseudoQuality}
\end{lemma}

Similarly to the proof of the global runtime, we wish to bound the number of greedy tournaments before either all released nodes in $N^{\log n}_\infty(u)$ become permanently passive or are deleted. 
We say that a node \textit{participates} in greedy tournament $t$ if it enters state $U'$ from $D'$ at time $t$. 
We say that greedy tournament $i$ is (locally) unclean, if there is a node v $\in N^{\log n}_\infty(u)$ in an active state of greedy tournament $i$ while $v$ is deleted or an edge is added adjacent to $v$.
Otherwise, greedy tournament $i$ is clean.
We show that if at least one node in $N^1_{\infty}(u)$ participates in greedy tournament $i$ and greedy tournament $i$ is (locally) clean, then with a constant probability, at least one node $v \in N^{\log n}_\infty(u)$ wins in greedy tournament $i$. 
Recall that winning indicates that $v$ enters state $W$.

\begin{lemma}
	Consider a clean greedy tournament $t$ where at least one node in $N^{1}_{\infty}(u)$ participates. Then there is some $v \in N^{\log n}_\infty (u)$ s.t. $v$ wins in greedy tournament $t$ with probability at least $1/6$.
	\label{lemma: pseudoWin}
\end{lemma}
\begin{proof}
	Let $N^i(u, t) \subseteq N^i_\infty(u)$ be the set of nodes (within $i$ hops from $u$) that participate in greedy tournament $t$.
To prove the claim, we first show that if $|N^i(u, t)| \leq 2|N^{i - 1}(u, t)|$ for any $i$, then there is a winning node in $N^i(u, t) \subseteq N^i_{\infty}(u)$ with probability at least $1/6$. 
Assume $|N^i(u, t)| \leq 2|N^{i - 1}(u, t)|$ and let $A(u, t)$ be the event that the maximum of $\{ X^u(t) \mid u \in N^{i}(u, t) \}$ is attained in a single node. 
Since greedy tournament $t$ is clean, the probability of $A(u, t)$ is at least $1/3$. 
As the coin tosses in $N^i(u, t)$ obey the same distribution, the maximum is attained equally likely in all the nodes. Therefore, a node from $N^{i - 1}(u, t)$ wins with probability
\[
	\frac{1}{3} \frac{|N^{i - 1}(u, t)|}{|N^i(u, t)|} \geq \frac{1}{6} \ .
\]

Assume now for contradiction that $|N^i(u, t)| > 2|N^{i - 1}(u, t)|$ for all $0 \leq i \leq \log n$. 
It follows that $|N^{\log n}(u, t)| > 2^{\log n} = n$, which is a contradiction. 
Therefore, for any clean greedy tournament $t$, there is always a node in $N^{\log n}(u, t) \subseteq N^{\log n}_{\infty}(u)$ that wins in greedy tournament $t$ with probability at least $1/6$.
\end{proof}

Now with an argument analogous to the one of Lemma~\ref{lemma: NumberGreedy}, we get that the expected number of clean greedy tournaments in which at least one node in $N^{1}_{\infty}(u)$ participates is $\bigO(C_u \log n)$. 
By applying Observation~\ref{obs: lenGreedy}, we get that the number of rounds in which any node in $N^{1}_{\infty}(u)$ participates in any greedy tournament is $\bigO(C_u \log^2 n)$.
Since there can be at most $2C_u$ unclean greedy tournaments, the number of rounds in which some node in $N^{1}_{\infty}(u)$ participates in a greedy tournament that is not clean is bounded by $C_u \cdot \bigO(\log n )$.
According to the design of the $D'$ state, there cannot be $2$ successive rounds in which $u$ is active in \greedy{} and no node in $N^{1}_{\infty}(u)$ participates in any greedy tournament or are in an active state of \Yuvals{}.
Finally, recalling that a non-affected node becomes passive in $\bigO(\log^2 n)$ rounds and by Corollary~\ref{cor: FairRunTime}, we get the following theorem.

\begin{theorem}
	The expected runtime is $\bigO((C_u + 1) \log^2 n)$ for any node $v$.
\end{theorem}
\section{Concluding remarks}
\label{section:conclusions}
As stated in \Sect{}~\ref{section:introduction}, in this paper, we take the first
step in the study of recovery from dynamic changes under the Stone Age model
and we believe that exploring this research direction further raises some
fascinating open questions:
How does the complexity of the MIS problem in dynamic graphs change under
various aspects of asynchrony (e.g., asynchronous wake-up or even fully
asynchronous message passing)?
Does there exist an efficient Stone Age MIS algorithm if the communication
model is modified so that the node receives its own transmissions (depicted by
allowing the graph to have self-loops)?
What would be the effect of augmenting the problem definition with
self-stabilization requirements?


\bibliographystyle{plain}
\bibliography{references}



\end{document}